\documentclass[prx, preprint, superscriptaddress]{revtex4-2}
\usepackage{graphicx}
\usepackage{amssymb}
\usepackage{amsfonts}
\usepackage{amsmath}
\usepackage{amsthm}
\usepackage{cancel}
\usepackage{tensor}
\usepackage{lmodern}
\usepackage{mathrsfs}
\usepackage{mathtools,slashed,isomath}
\usepackage{mathdots}
\usepackage{collref}
\usepackage[colorlinks=true, citecolor=blue, urlcolor=blue]{hyperref}
\usepackage{tikz}
\usepackage{bm}
\usepackage{hyperref}
\hypersetup{colorlinks=true,linkcolor=blue,urlcolor=blue}
\usepackage{lineno}
\usepackage[markup=underlined]{changes}
\makeatletter
\AddToHook{cmd/added/before}{\def\Changes@AuthorColor{blue}}
\AddToHook{cmd/deleted/before}{\def\Changes@AuthorColor{brown}}
\AddToHook{cmd/replaced/before}{\def\Changes@AuthorColor{purple}}
\makeatother
\usepackage{todonotes}
\setcommentmarkup{\todo[color={pink},size=\scriptsize]{#3: #1}}

\usetikzlibrary{shapes.geometric}
\newtheorem{theorem}{Theorem}

\newcommand{\beginmethod}{%
       \setcounter{section}{0}
        \renewcommand{\thesection}{\Roman{section}}%
     }

\makeatletter 
\def\maketitle{
\@author@finish
\title@column\titleblock@produce
\suppressfloats[t]}
\makeatother

\begin{document}
\title{Structures in higher-order quantum correlations due to non-spatial symmetries}
\author{Li Sun}
\author{Chong Chen}
\affiliation{Department of Physics, The Chinese University of Hong Kong, Shatin, New Territories, Hong Kong, China}%
\affiliation{The State Key Laboratory of  Quantum Information Technologies and Materials, The Chinese University of Hong Kong, Shatin, New Territories, Hong Kong, China}
\affiliation{New Cornerstone Science Laboratory, The Chinese University of Hong Kong, Shatin, New Territories, Hong Kong, China}%

\author{Ren-Bao Liu}
\email{rbliu@cuhk.edu.hk}
\affiliation{Department of Physics, The Chinese University of Hong Kong, Shatin, New Territories, Hong Kong, China}%
\affiliation{The State Key Laboratory of  Quantum Information Technologies and Materials, The Chinese University of Hong Kong, Shatin, New Territories, Hong Kong, China}%
\affiliation{New Cornerstone Science Laboratory, The Chinese University of Hong Kong, Shatin, New Territories, Hong Kong, China}%
\affiliation{Centre for Quantum Coherence, The Chinese University of Hong Kong, Shatin, New Territories, Hong Kong, China}%

\begin{abstract} 
 Quantum nonlinear spectroscopy (QNS) via a quantum sensor can access $2^{n-1}$ types of $n$-th order contour-time-ordered correlations (CTOCs) arising from different orderings of quantum operators, while classical nonlinear spectroscopy can detect only one in each order.  QNS and its classical counterpart have similar spatial symmetry properties, but they are expected to have characteristically different non-spatial symmetry properties since different orderings of operators can behave differently under non-spatial transformations (such as exchange of operators). Here, we investigate how higher-order correlations extracted by QNS are constrained by non-spatial symmetries, including particle-hole (C), time-reversal (T), chiral (S) symmetry, and time translation symmetry. We find that the generalized C-symmetry imposes special selection rules on QNS, and the generalized T- and S-symmetry relate CTOCs to out-of-time-order correlations (OTOCs). The time translation symmetry leads to a generalized fluctuation-dissipation theorem for the spectra of higher-order CTOCs and OTOCs. This work discloses deep structures in higher-order quantum correlations due to non-spatial symmetries and provides access to certain types of OTOCs that are not directly observable.
\end{abstract}

\maketitle
\newpage

\section*{Introduction}
\noindent
 Quantum correlations between physical quantities reflect the dynamics and structure of quantum systems. The study of quantum correlations focuses mainly on first- and second-order correlations since they are sufficient to construct the higher-order ones when the fluctuations have Gaussian statistics, which is often the case in a macroscopic system. Recently, the importance of high-order quantum correlations is revealed in the study of finite-size many-body systems~\cite{schweigler2017experimental}, open quantum systems~\cite{PhysRevLett.123.050603}, quantum sensing~\cite{laraoui2013high, pfender2019high, wang2021classical, PhysRevLett.130.070802}, quantum chaos and information scrambling~\cite{hosur2016chaos, PhysRevX.8.021013, PhysRevLett.126.030601, PhysRevA.103.062214,PhysRevA.94.040302, PhysRevLett.124.200504}, 
 and quantum foundation~\cite{tura2015nonlocality, PhysRevA.98.022116}. Higher-order quantum correlations have rich structures. In particular, because quantities in quantum mechanics in general do not commute, the orderings of the quantities result in different correlations. For example,
 correlations of physical quantities at different times can be classified into contour-time-ordered correlations (CTOCs) and out-of-time-ordered correlations (OTOCs), depending on whether the time ordering is violated~\cite{haehl2019classification}. OTOCs are important in theories of quantum information scrambling and quantum chaos in many-body systems~\cite{hosur2016chaos, PhysRevX.8.021013, PhysRevLett.126.030601, PhysRevA.103.062214,PhysRevA.94.040302, PhysRevLett.124.200504}.
 Furthermore, the quantities in the correlations can form different concatenations of commutators or anti-commutators~\cite{PhysRevLett.123.050603}.

Nonlinear spectroscopy is a widely used approach to extracting correlations. However, it can access only one type of CTOCs in each order, namely, $\left\langle i^{n} \left[\left[\left[\hat{B}_{n+1},\hat{B}_{n}\right],\ldots,\hat{B}_2\right], \hat{B}_1\right]\right\rangle$ (which involves only time-ordered commutators between the operators) in the $n$-th order nonlinear response. This limitation is due to the fact that under the driving of a ``classical force'' $f\left(t\right)$ coupled to a ``displacement'' operator $\hat{B}\left(t\right)$, the evolution (and hence response) of a quantum system (described by a state $\hat\rho$) is governed by the commutator $-i\left[f\left(t\right)\hat{B}\left(t\right),\hat{\rho}\right]=-if\left(t\right)\left[\hat{B}\left(t\right),\hat{\rho}\right]$. Noise spectroscopy~\cite{crooker2004spectroscopy, liu2010dynamics,sinitsyn2016theory}
 is another method to measure the CTOCs, but it can access only those involving solely anti-commutators (corresponding to correlations of classical noises). In the second order, there are only two types of correlations, corresponding to the commutator and the anti-commutator of two quantities, which can be extracted using linear response and noise spectroscopy, respectively. However, in general, conventional nonlinear spectroscopy and noise spectroscopy are insufficient to extract the full set of CTOCs, which has $2^{n-1}$ different types in the $n$-th order corresponding to different concatenations of commutators and anti-commutators of $n$ operators. It has recently been demonstrated that the quantum sensing approach can be used to extract the entire set of CTOCs, which is called quantum nonlinear spectroscopy (QNS)~\cite{PhysRevLett.123.050603, PhysRevLett.132.200802, meinel2022quantum}.
However, there are still $\left(n!-2^{n-1}\right)$ types of $n$-th order OTOCs  left beyond direct access because the dynamics of any physical system fulfills the causality and therefore involves only CTOCs. The (approximate) extraction of OTOCs requires (approximately) time-reversed evolution, which imposes a substantial experimental challenge for studying many-body systems~\cite{PhysRevX.7.031011,garttner2017measuring, PhysRevX.9.021061, PhysRevLett.128.140601, PhysRevA.108.062608}.

Symmetry is a cornerstone of physics. The consideration of symmetry not only simplifies the analysis of systems, but also constrains physical theories.
Symmetry analysis is particularly useful in nonlinear spectroscopy
since the higher-order responses of a system to external perturbations are complex due to the existence of many possible spatial configurations of the external forces and the responses~\cite{shen1984principles}. 
Especially for solid-state systems, the point groups, which describe the crystallographic symmetries, simplify nonlinear spectroscopy by selection rules~\cite{shen1984principles}.

QNS, like its classical counterpart, have structures arising from symmetries. Spatial symmetries have similar effects on quantum and classical nonlinear spectroscopy. But QNS is expected to have distinct structures due to non-spatial symmetries. For example, under matrix transpose or complex conjugation transforms, the correlations $\left\langle\frac{1}{2}\left(\hat{A}\hat{B}+\hat{B}\hat{A}\right)\right\rangle$ and $\left\langle i\left(\hat{A}\hat{B}-\hat{B}\hat{A}\right)\right\rangle$ behave differently. Therefore, non-spatial symmetries, such as time reversal (T), particle-hole (C), and their combination, chiral symmetry (S), and time translation invariance, are expected to induce unique structures in QNS, or, more generally, in higher-order quantum correlations. The T-, C-, and S-symmetries rigorously categorize Hamiltonians within the tenfold internal symmetry classes~\cite{RevModPhys.88.035005}. They are useful in classifying Anderson transitions~\cite{RevModPhys.80.1355} and topological insulators and superconductors~\cite{RevModPhys.88.035005}, and in characterizing quantum chaos~\cite{haake1991quantum}. The time-translation invariance leads to the fluctuation-dissipation theorem (FDT) between the two types of second-order correlations in systems at thermal equilibrium, which establishes a profound connection between the fluctuations and the response or dissipation of the systems~\cite{kubo1966fluctuation, PhysRev.83.34}.
However, the structures in higher-order quantum correlations due to the discrete non-spatial symmetries (C, T, and S) have not been charted, and no FDT-like relations are known beyond the second-order correlations except for classical systems~\cite{PhysRevLett.87.040601, PhysRevA.86.043818} and for a special class of OTOCs in quantum systems~\cite{PhysRevE.97.012101}.

Here, we present the discovery of the structures in higher-order quantum correlations resulting from non-spatial symmetries. We obtain selection rules of QNS due to C-symmetry (Theorem~\ref{theorem_C}), which can reduce the number of CTOCs (among all $2^{n-1}$ types of $n$-th order CTOCs) that need to be measured. By Theorem~\ref{theorem_T} and~\ref{theorem_S}, we demonstrate that T- and S-symmetry induce connections between OTOCs of different ranks (including the rank-1 OTOCs, i.e., CTOCs), which may enable the measurement of certain types of rank-2 OTOCs of T- or S-symmetric systems using QNS without the requirement of (unphysical) time-arrow reversal. For time translation invariant systems, we establish a generalized FDT for higher-order quantum correlations, with the original FDT being a special case in the second order.  We find that the FDT-like relations in higher-order correlations involve both CTOCs and OTOCs, which allows us to obtain the spectra of OTOCs from those of CTOCs. We demonstrate these structures arising from the non-spatial symmetries using the transverse-field Ising model (TFIM) as a case study.

\section*{Results}
\subsection*{Classification of high-order quantum correlations}
The correlations of $n$ physical operators can be expanded in the $n!$ Wightman correlations~\cite{haehl2019classification}
\begin{align}\label{eq:Wightman}
    W_{{\sigma},n} \equiv \operatorname{Tr}\left[\hat{B}_{\sigma\left(n\right)} \hat{B}_{\sigma\left(n-1\right)} \cdots \hat{B}_{\sigma\left(1\right)} \hat{\rho}\right],
\end{align}
where $\hat{B}_{k_i}\equiv \hat{B}_{k_i}\left(t_i\right)$ is an operator at time $t_i$ in the interaction picture, $\sigma$ denotes a permutation of $\left(1,2,\ldots,n\right)$,  and $\hat{\rho}$ is the initial state of the system.
By linear combination, the Wightman correlations in Eq.~\eqref{eq:Wightman} can be rewritten in a more physical form as
\begin{align}\label{eq:physical}
    C^{\boldsymbol \eta}_{\sigma,n}={\rm Tr}{\left[{\mathbb B}^{+}_{\sigma\left(n\right)}{\mathbb B}^{\eta_{n-1}}_{\sigma\left(n-1\right)}\cdots{\mathbb B}^{\eta_1}_{\sigma\left(1\right)}\hat{\rho}\right]},
\end{align}
where the superscripts $\eta_{i}=\pm$ and the superoperators are defined as  ${\mathbb B}^+ \hat{A}\equiv\frac{1}{2}\{\hat{B}, \hat{A}\}$ (essentially an anti-commutator) and ${\mathbb B}^{-} \hat{A}\equiv -i\left[\hat{B}, \hat{A}\right]$ (essentially a commutator). The term ${\mathbb B}^{-}_i\hat{\rho}$ represents the evolution of the system driven by a force coupled to $\hat{B}_i$, and ${\mathbb B}^{+}_i\hat{\rho}$ has the physical meaning of the noise from the system as indicated by the fact that ${\rm Tr}\left[{\mathbb B}^{+}_i\hat{\rho} \right] = \langle {B}_i\rangle$. 
The dynamics of a system is fully characterized by the so-called CTOCs, in which the operators are contour-time-ordered [see Fig.~\ref{fig:CTOC} (a)]. As can be seen from Fig.~\ref{fig:CTOC} (a), the trace in Eq.~\eqref{eq:Wightman}  is the same whether the last operator $B_n$ appears on the forward or backward branch, so there are $2^{n-1}$ types of $n$-th order CTOCs. 
A correlation in the physical form $ C^{\boldsymbol \eta}_{\sigma,n}$ is a CTOC if $t_{\sigma(1)}\le t_{\sigma(2)}\le\cdots\le t_{\sigma(n)}$.
In particular, all second-order correlations are CTOCs, with  $C^{+-}_{2,1}=-i\left\langle\left[\hat{B}_2,\hat{B}_1\right]\right\rangle$ describing the propagation dynamics and $C^{++}_{2,1}=\left\langle\frac{1}{2}\left\{\hat{B}_2,\hat{B}_1\right\}\right\rangle$ the fluctuation correlation. Note that these two 2nd-order CTOCs, through linear combination, can be transformed to the advanced and retarded Green's functions or the Keldysh-Kadanoff-Baym Green's functions~\cite{economou2006green}.   In addition to CTOCs, there are $n!-2^{n-1}$ types of $n$-th order correlations in which the operators appear in more than one contours of time, called rank-$k$ OTOCs if $k$ time-contours are needed [see Fig.~\ref{fig:CTOC}{(b)} for an example of rank-3 OTOC]~\cite{haehl2019classification}.

\begin{figure}
 \centering
 \includegraphics[width=0.40\textwidth]{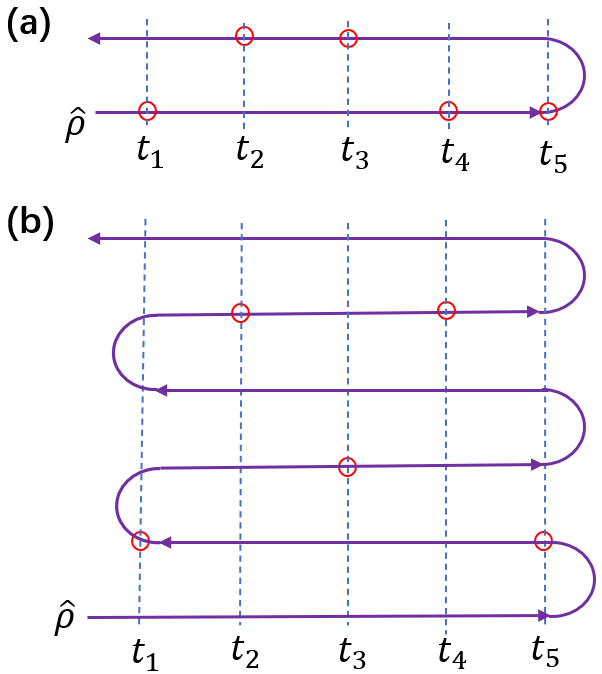}
 \caption{\textbf{Examples of the 5-th order correlations}. (a) CTOC ${\rm Tr}\left[\hat{B}_2\hat{B}_3\hat{B}_5\hat{B}_4\hat{B}_1\hat{\rho}\right]$ (i.e., $ {\rm Tr}\left[\hat{B}_5\hat{B}_4\hat{B}_1\hat{\rho}\hat{B}_2\hat{B}_3\right]$ by cyclic permutation in the trace), with a time-ordered sequence of operators $\hat{B}_5\hat{B}_4\hat{B}_1\hat{\rho}$ on the forward branch of the time-contour and an anti-time ordered sequence $\hat{B}_2\hat{B}_3$ on the backward branch.  {(b) Rank-3 OTOC ${\rm Tr}\left[\hat{B}_4\hat{B}_2\hat{B}_3\hat{B}_1\hat{B}_5\hat{\rho}\right]$, which needs three time-contours to arrange the operators. Here we set $t_1<t_2<\cdots <t_5$.}}
\label{fig:CTOC}
\end{figure}

\subsection*{Structures in higher-order correlations due to C-, T-, and S-symmetry}
A system is regarded to have {\em generalized} C-, T-, and S-symmetry, if the {\em generalized} particle-hole exchange operation (${\mathcal C}$), time reversal (${\mathcal T}$), or their combination (${\mathcal S}\equiv{\mathcal C}{\mathcal T}$) correspondingly transfers its Hamiltonian $H$ and initial state $\rho$, both being Hermitian matrices in a given basis, such that~\cite{PhysRevB.55.1142, RevModPhys.87.1037, RevModPhys.88.035005, PRXQuantum.4.040312} 
\begin{subequations}
\begin{align}
    {\mathcal C} H^T {\mathcal C}^{-1}  =- {H},  & \ \ {{\mathcal C}}  { \rho}^T  {{\mathcal C}} ^{-1} = {\rho};
    \label{eq:C-transform}
    \\
 {{\mathcal T}}  { H}^{\ast} {{\mathcal T}}^{-1}  =+{H}, & \ \ {{\mathcal T}}  { \rho}^{\ast}  {{\mathcal T}} ^{-1} = {\rho};
  \label{eq:T-transform}
 \\
  { {\mathcal S}} {H} {{\mathcal S}}^{-1}  =-{H}, & \ \ {{\mathcal S}}  { \rho}  {{\mathcal S}} ^{-1} = {\rho}.
\label{eq:S-transform}    
\end{align}
\end{subequations}
Here, the symmetries are called ``generalized'' in the sense that the transforms ${\mathcal C}$, ${\mathcal T}$, and ${\mathcal S}$ are only required to be invertible and do not need to be unitary matrices or to satisfy the conditions ${\mathcal C} {\mathcal C}^{\ast}=\pm I$, ${\mathcal T} {\mathcal T}^{\ast}=\pm I$, and ${\mathcal S}^2= I$~\cite{PhysRevB.55.1142, RevModPhys.87.1037, RevModPhys.88.035005, PRXQuantum.4.040312}. 
In particular, the generalized T-symmetry here does not lead to Kramers' degeneracy.

\begin{figure}
 \centering
 \includegraphics[width=0.8\textwidth]{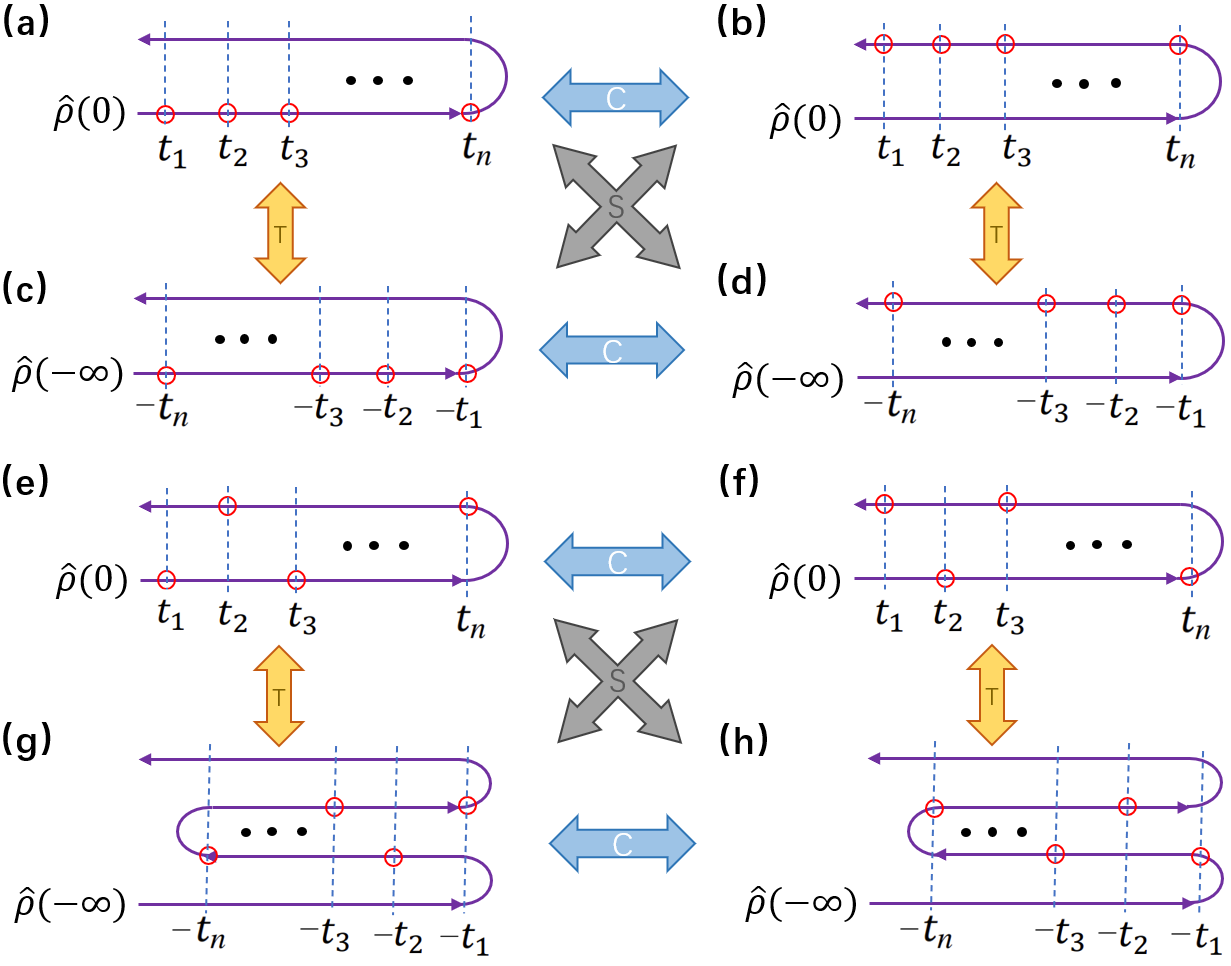}
 \caption{\textbf{Wightman correlations and their relations under ${\mathcal C}$-, ${\mathcal T}$-, and ${\mathcal S}$-transforms.} (a-d) CTOCs for operators in ascending and descending time sequences. (e,f) General rank-1 OTOCs (i.e., CTOCs). (g, h) Rank-2 OTOCs.  Correlations in (a/c/e/g) are correspondingly related to those in (b/d/f/h) by ${\mathcal C}$-transform,  (a/b/e/f) to (c/d/g/h) by ${\mathcal T}$-transform, and  (a/b/e/f) to (d/c/h/g) by ${\mathcal S}$-transform, as indicated by the arrows marked by C, T, and S, correspondingly.}
\label{fig:W_CST}
\end{figure}

The generalized C-symmetry leads to selection rules for high-order quantum correlations as summarized in the following theorem.
\begin{theorem}\label{theorem_C}
    When the system is C-symmetric and the physical quantities of interest are C-symmetric/anti-symmetric, i.e., ${\mathcal C}   { B}_{k_i} ^T {\mathcal C} ^{-1} =\alpha_i { B}_{k_i}$ for $\alpha_i=+/- 1$, respectively, the correlations in the Wightman basis satisfy 
\begin{align}\label{eq:theorem_C}
        W_{\tilde{{\sigma}},n}= W_{\sigma,n} \prod_{i=1}^n \alpha_i ,
    \end{align}
    where $\tilde{\sigma}\left(i\right)\equiv\sigma\left(n-i+1\right)$ denotes the reversed order of $\sigma$, and the physical correlations satisfy the selection rule
    \begin{equation}\label{eq:theorem_C_selection}
            C^{\boldsymbol \eta}_{\sigma, n}=0\ \  {\rm if }\ \  
            \prod_{j=1}^n\eta_j+\prod_{j=1}^n \alpha_j=0. 
    \end{equation}
\end{theorem}
\begin{proof}
Using the identity ${\rm Tr}\left[{A}\right]={\rm Tr}\left[{A}^T\right]$, we get
\begin{align}\label{eq:W_transposed}
    W_{\sigma,n} =\operatorname{Tr}\left[ {B}^T_{\sigma\left(1\right)} {B}^T_{\sigma\left(2\right)} \cdots  {B}^T_{\sigma\left(n\right)}    {\rho}^T\right],
\end{align}
where $B_j^T=e^{-iH^Tt_j}B^T_{k_j}e^{iH^Tt_j}$.  It is straightforward to show that ${{\mathcal C}}   {B}_{j} ^T {{\mathcal C}}^{-1}=\alpha_j {B}_{j}$ when Eq.~\eqref{eq:C-transform} is satisfied. Inserting ${\mathcal C}^{-1}{\mathcal C}$ between operators in  Eq.~\eqref{eq:W_transposed} and using ${{\mathcal C}}   {B}_{j} ^T {{\mathcal C}}^{-1}=\alpha_j {B}_{j}$, we  get Eq.~\eqref{eq:theorem_C}. The physical correlations $C^{\boldsymbol \eta}_{\sigma,n}$ are linear combinations of the Wightman correlations with the expansion coefficients associated with $W_{\sigma,n}$ and $W_{\tilde{\sigma},n}$ being the same/opposite when  $\prod_{j=1}^n\eta_j=+/-1$, which leads to the selection rule in Eq.~\eqref{eq:theorem_C_selection}.
 \end{proof}

 \begin{table}
    \centering
    \begin{tabular}{c|c|c|c|c|c|c}
    \hline
    \hline
    C-symmetry   & $C^{+-}_{\sigma,{2}} $ & $C^{++}_{\sigma,{2}} $ & $C^{+--}_{\sigma,{3}} $ &  $C^{++-}_{\sigma,{3}} $&  $C^{+-+}_{\sigma,{3}} $ & $C^{+++}_{\sigma,{3}} $  \\
      \hline
    ~~$B_{k_i}^T \rightarrow - B_{k_i}$ ~~ & $0$  &   -   & $0$ & - & - & $0$ \\
      \hline
    ~~$B_{k_i}^T \rightarrow + B_{k_i}$ ~~ &  $0$  &  -   &-  & $0$& $0$& -  \\
      \hline
    \end{tabular}
    \caption{\textbf{The selection rule induced by C-symmetry in second- and third-order CTOCs.} $0$ means the correlation must be zero.}
    \label{tab1}
\end{table}

The relation in Eq.~\eqref{eq:theorem_C} is shown in Fig.~\ref{fig:W_CST} (indicated by the arrows marked by C). Note that the ${\mathcal C}$-transform of the correlations only switches the operators from the {$j$-th} forward/backward branch to {the $\left(n-j+1\right)$-th backward/forward}  one in the time-contour(s) without changing the rank of the OTOCs. The C-symmetry-induced selection rule in second- and third-order CTOCs is shown in Table~\ref{tab1}.

The generalized T-symmetry leads to the following theorem.
\begin{theorem}\label{theorem_T}
    When the system is T-symmetric and the physical quantities of interest are T-symmetric/anti-symmetric, i.e., ${\mathcal T}   { B}_{k_j}^{\ast} {\mathcal T} ^{-1} =\beta_j { B}_{k_j}$ for $\beta_j=+/- 1$, respectively, and the system is initially in a stationary state, that is, ${\rho}\left(0\right)=\rho\left(t<0\right)=\rho\left(-\infty\right)$, the Wightman correlations are related by  
\begin{align}
    W_{\sigma,n} & = W_{\tilde{{\sigma}},n}\left(\{t_j\rightarrow -t_j\}\right) \prod_{j=1}^n\beta_j. \label{eq:theorem_T} 
\end{align}
    where $\{t_j\rightarrow -t_j\}$ means that each $t_j$ is replaced by $-t_j$ in $W_{\tilde{\sigma},n}$. 
\end{theorem}
\begin{proof}
    Inserting ${\mathcal T}^{-1}{\mathcal T}$ between operators in Eq.~\eqref{eq:W_transposed} and using ${{\mathcal T}} e^{-iH^T t}{{\mathcal T}}^{-1} ={{\mathcal T}} e^{-iH^* t}{{\mathcal T}}^{-1}= e^{-iHt}$ and hence ${{\mathcal T}} B_{k_j}^T\left(t_j\right){{\mathcal T}}^{-1}={{\mathcal T}} B_{k_j}^{\ast}\left(t_i\right){{\mathcal T}}^{-1}=\beta_j B_{k_j}\left(-t_j\right)$ (for $B_{k_j}$ being Hermitian matrices), we get Eq.~\eqref{eq:theorem_T}.
\end{proof}
 
 According to Eq.~\eqref{eq:theorem_T}, T-symmetry transforms the operators from the {$j$-th} forward/backward branch in a rank-$k$ correlation to the $\left(k-j+1\right)$-th backward/forward one in the contour(s) and simultaneously changes $t_i$ to $-t_i$, as shown in Fig.~\ref{fig:W_CST}(c, d, g, h). After the ${\mathcal T}$-transform, the time where the operators $B_j$ occurs is $-t_j <0$ [see Fig.~\ref{fig:OTOC_TS}(a)]. With the assumption that the system is initially in the stationary state such that ${\rho}\left(0\right)=\rho\left(t<0\right)=\rho\left(-\infty\right)$, we can move the state to a time $t_0<-t_n$ and add an empty forward (backward) branch at the beginning (end) of the time-contour to ensure that the contour always starts forward and ends backward in the initial state ${\rho}$ [see Fig.~\ref{fig:OTOC_TS}(a)]. Thus, we obtain 
 $W_{\tilde{{\sigma}},n}\left(\{t_j\rightarrow -t_j\}\right)\equiv {\rm Tr}\left[B_{\sigma\left(1\right)}\left(-t_{\sigma \left(1\right)}\right)B_{\sigma\left(2\right)}\left(-t_{\sigma\left(2\right)}\right)\cdots B_{\sigma\left(n\right)}\left(-t_{\sigma\left(n\right)}\right)\rho\left(-\infty\right)\right]$.

Interestingly, some types of CTOCs can be related to rank-2 OTOCs under the ${\mathcal T}$-transform. For example, the CTOCs in Fig.~\ref{fig:W_CST}(e, f) are transformed to rank-2 OTOCs in (g, h), respectively. In general, OTOCs of different ranks can be converted as described below.
\begin{enumerate}
\item If a rank-$k$ OTOC contains {a sequence of} operators in the first forward (last backward) branch of time-contour but no operators in the last backward (first forward) branch, it will be converted by time reversal to a rank-$k$ OTOC with {the} sequence of operators in the first forward (last backward) branch {reversely ordered} and {no} operators in the last backward (first forward) branch. Examples are shown in Fig.~\ref{fig:W_CST}(a/b)~$\leftrightarrow$~(c/d) for rank-1~$\stackrel{\rm T}{\Longleftrightarrow}$~rank-1 conversion and in Fig.~\ref{fig:OTOC_TS}(b) for rank-2~$\stackrel{\rm T}{\Longleftrightarrow}$~rank-2.
\item  If a rank-$k$ OTOC contains operators in both the first forward and the last backward branches of time-contour, it will be converted by time reversal to a rank-$\left(k+1\right)$ OTOC with no operators in {either} the first forward or the last backward branch, and vice versa. Examples are shown in Fig.~\ref{fig:W_CST}(e/f)$\leftrightarrow$(g/h) for rank-1~$\stackrel{\rm T}{\Longleftrightarrow}$~rank-2 and in Fig.~\ref{fig:OTOC_TS}(c) for rank-2~$\stackrel{\rm T}{\Longleftrightarrow}$~rank-3.

\end{enumerate}

\begin{figure}
 \centering
 \includegraphics[width=0.95\textwidth]{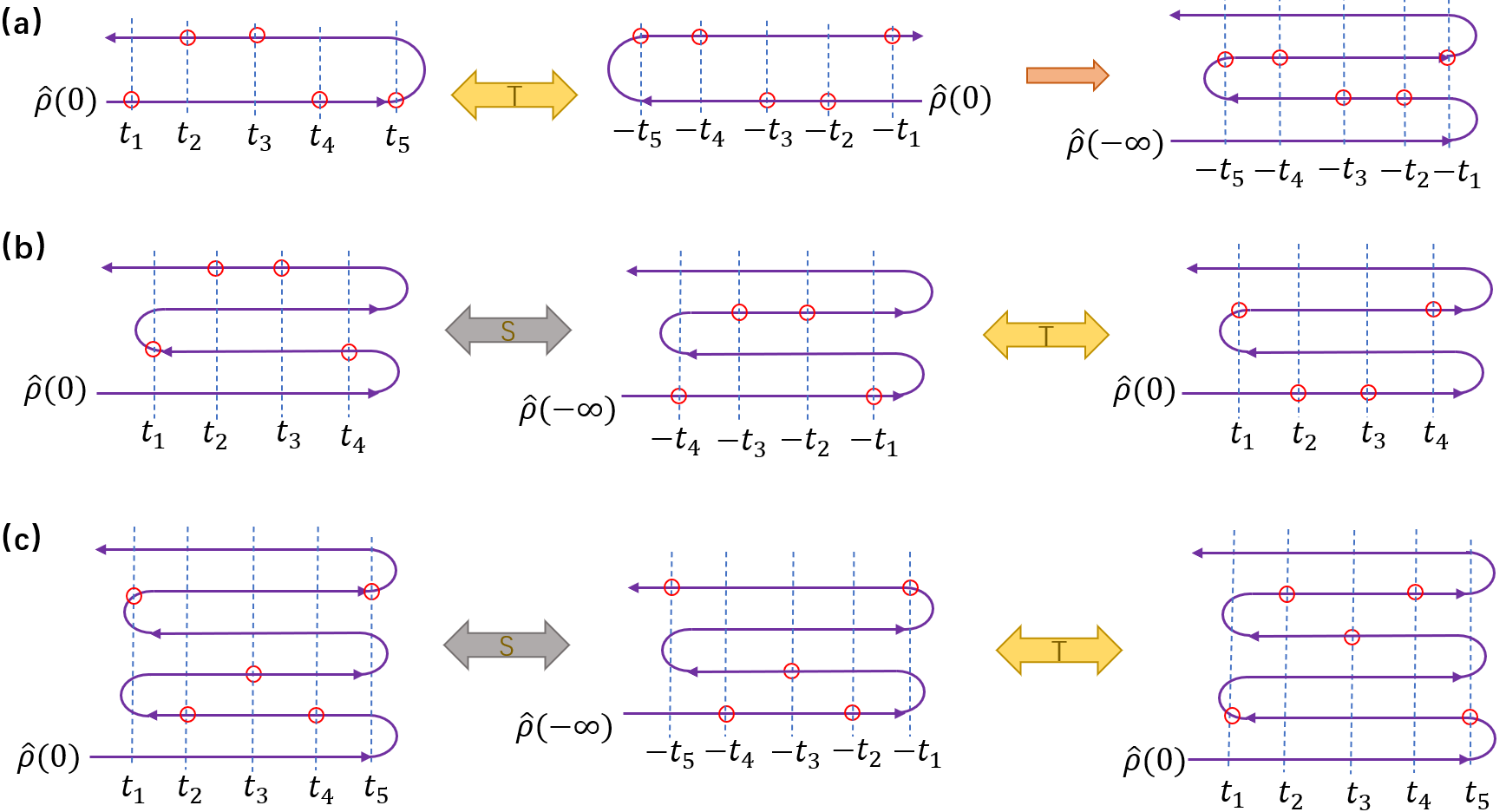}
 \caption{\textbf{OTOCs under ${\mathcal T}$- and ${\mathcal S}$-transforms.} (a) A 5th-order CTOC is transformed to a rank-2 OTOC. When there are operators on both the first forward and the last backward branches of the time-contour, after the ${\mathcal T}$-transform, two empty branches need to be added to ensure the time-contour always starts {forward} from  and ends backward to the system state, which is assumed to be stationary so that ${\rho}\left(0\right)={\rho}\left(-\infty\right)$. (b) Conversion between rank-2 OTOCs by ${\mathcal T}$- and ${\mathcal S}$-transforms. (c) Conversion between rank-2 and rank-3 OTOCs by ${\mathcal T}$- and ${\mathcal S}$-transforms.}
\label{fig:OTOC_TS}
\end{figure}

The connections among high-order quantum correlations arising from generalized S-symmetry are summarized in the following theorem.
\begin{theorem}\label{theorem_S}
    When the system is S-symmetric and the physical quantities of interest are S-symmetric/anti-symmetric, i.e., ${\mathcal S}   { B}_{k_i}  {\mathcal S} ^{-1} =\gamma_i { B}_{k_i}$ for $\gamma_i=+/- 1$, respectively, and the system is initially in a stationary state, that is, $\rho\left(0\right)=\rho\left(t<0\right)=\rho\left(-\infty\right)$, the correlations in the Wightman basis are related by 
    \begin{equation}\label{eq:theorem_S}
    W_{\sigma,n}=W_{\sigma,n}\left(\{t_i\rightarrow-t_i\}\right) \prod_{i=1}^n\gamma_i,
    \end{equation}
    where $\{t_i\rightarrow-t_i\}$ means that each $t_i$ is replaced by $-t_i$ in  $W_{\sigma,n}$.
\end{theorem}
\begin{proof}
Inserting ${\mathcal S}^{-1}{\mathcal S}$ between operators in Eq.~\eqref{eq:Wightman} and using ${\mathcal S}B_{k_j}\left(t_j\right){\mathcal S}^{-1}=\gamma_j B_{k_j}\left(-t_j\right)$, we get Eq.~\eqref{eq:theorem_S}.
\end{proof}

Similarly to the ${\mathcal T}$-transform, when the system is initially in a stationary state, the ${\mathcal S}$-transform also relates certain rank-$k$ OTOCs to rank-$k$, rank-$\left(k-1\right)$, or rank-$\left(k+1\right)$ OTOCs, depending on whether or not there are operators in the first forward and/or the last backward branch of the time-contour. Examples of conversion between OTCOS of different ranks are shown in Fig.~\ref{fig:W_CST}(a/b)~${\leftrightarrow}$~(d/c) for rank-1 $\stackrel{\rm S}{\Longleftrightarrow}$ rank-1, Fig.~\ref{fig:W_CST}(e/f)~${\leftrightarrow}$~(h/g) for rank-1~$\stackrel{\rm S}{\Longleftrightarrow}$~rank-2, Fig.~\ref{fig:OTOC_TS}(b) for rank-2~$\stackrel{\rm S}{\Longleftrightarrow}$~rank-2 conversion, and  Fig.~\ref{fig:OTOC_TS}(c) for rank-3~$\stackrel{\rm S}{\Longleftrightarrow}$~rank-2.

\begin{figure}
 \centering
 \includegraphics[width=0.45\textwidth]{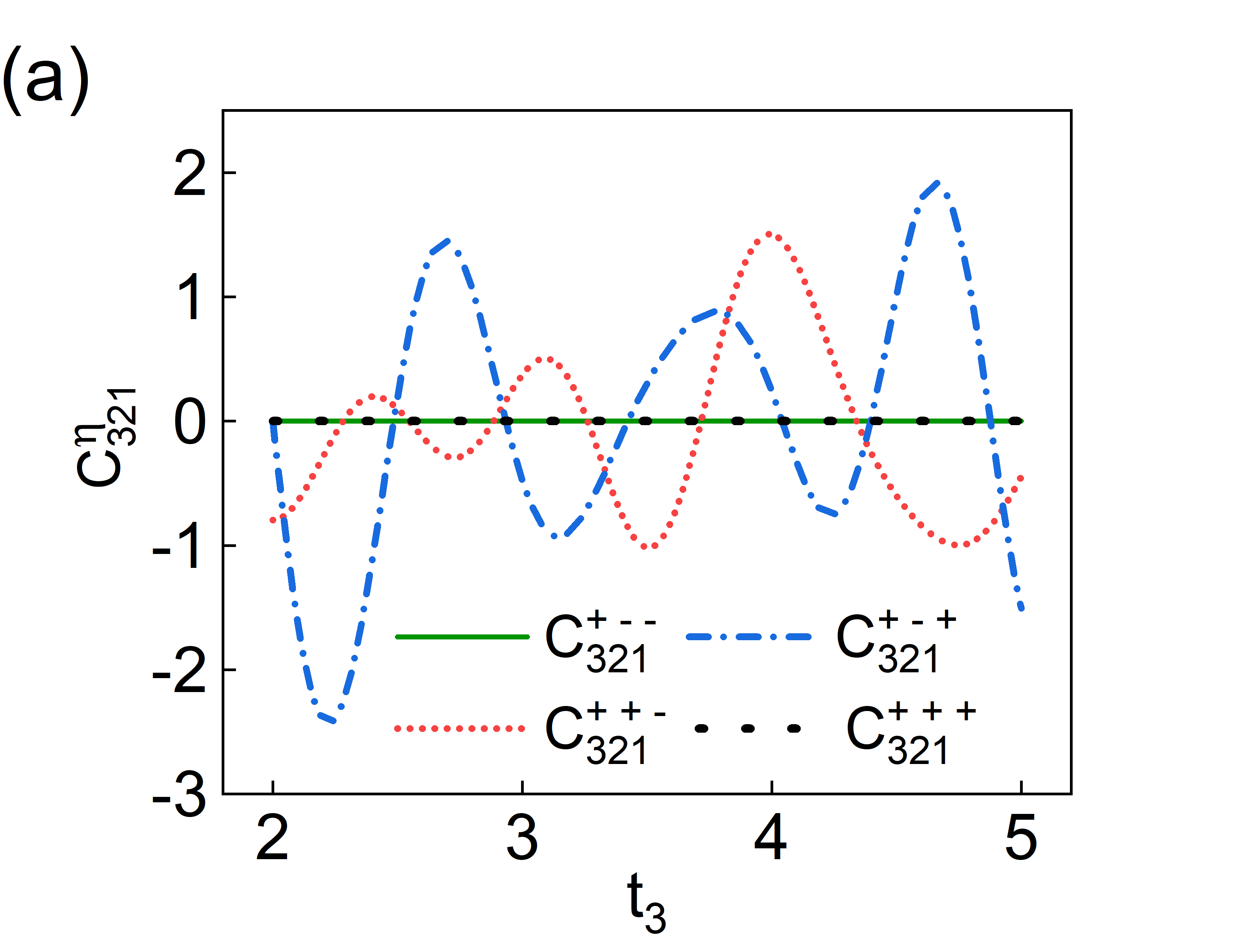}
 \includegraphics[width=0.45\textwidth]{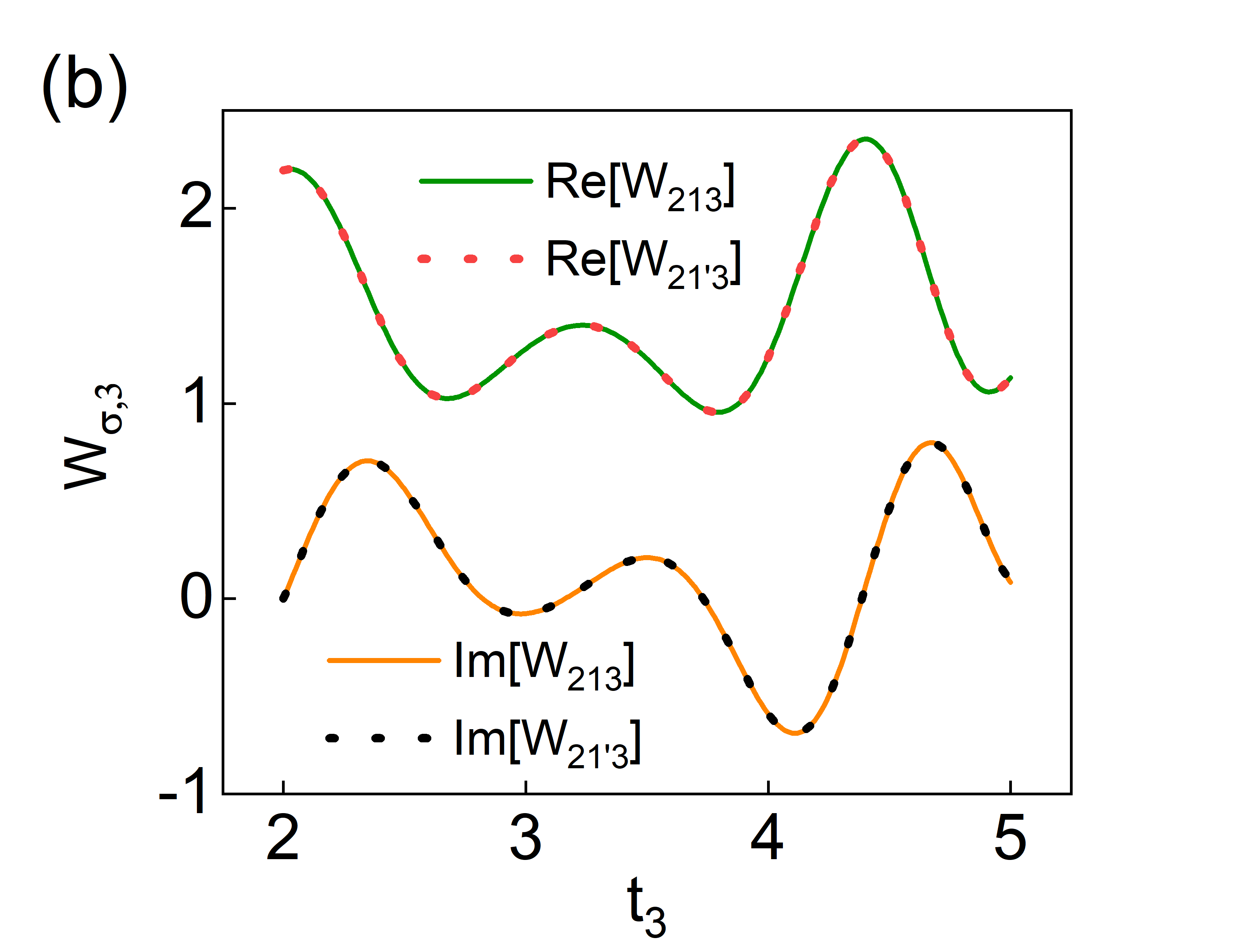}
 \caption{\textbf{Third-order correlations in transverse-field Ising model.}(a) Third-order CTOCs $C^{\boldsymbol \eta}_{321}$ of C-anti-symmetric observables as functions of $t_3$. (b) Real and imaginary parts of Wightman correlations $W_{213}$ and $W_{21^{\prime}3}$ of T-symmetric observables as functions of $t_3$ (with $t^{\prime}_1\equiv t_3+t_2-t_1$). In the calculation, the periodic boundary condition is chosen and parameters are such that $\lambda=1.5$, $\beta=1$, $N=8$, $t_1=0$ and $t_2=2$.} 
\label{fig2}
\end{figure} 

We consider the correlations in the TFIM~\cite{mattis2006theory}
to demonstrate the above theorems. The Hamiltonian {reads}
\begin{equation}
{H} = -\sum^{N}_{j=1} {\sigma}_j^z - \lambda \sum^{N}_{j=1} {\sigma}_j^x {\sigma}_{j+1}^x,
\label{eq:TFIM}
\end{equation}
where ${\sigma}^{x/y/z}_j$ are the Pauli matrices of the $j$-th spin along the $x/y/z$ axes with the basis chosen such that ${\sigma}^z_{j}$ is diagonal and ${\sigma}^x_{j}$ is real.  With the definitions ${\mathcal C} \equiv \prod^{N/2}_{l=1} \left( {\sigma}^x_{2l-1} {\sigma}^y_{2l} \right)$
and ${\mathcal T} \equiv \prod^{N}_{l=1} {\sigma}^z_{l}$, the Hamiltonian has both C- and T-symmetry.
We focus on third-order correlations. To demonstrate the effects of C-symmetry, we consider a C-symmetric initial state ${\rho} = \prod^{N/2}_{l=1} \left[  \left(\sigma_{2l-1}^x+I_{2l-1}\right) {I}_{2l}/4\right]$ and choose a C-anti-symmetric observable ${B} = \frac{1}{\sqrt{N}} \sum_j {\sigma}^z_j$.  Numerical results [Fig.~\ref{fig2}(a)] are consistent with the selection rules for the third-order CTOCs in Table~\ref{tab1}, 
with all terms but $C^{+-+}_{321}$ and $C^{++-}_{321}$ being zero (hereafter we denote a permutation $\sigma$ using the resultant sequence after the permutation, such as $321$). To show the effects of T-symmetry, we consider the equilibrium state $\rho\left(0\right)= \exp \left(-\beta H\right)/Z=\rho\left({-\infty}\right)$, where $Z={\rm Tr}\exp{\left(-\beta {H}\right)}$ with $\beta$ being the inverse temperature. Both the state $\rho$ and the observable $B$ are T-symmetric. As shown in Fig.~\ref{fig2}(b), the rank-2 OTOC $W_{213}\equiv {\rm Tr}\left[{B}\left(t_2\right){B}\left(t_1\right){B}\left(t_3\right)\rho\left(0\right)\right]$ and the CTOC  $W_{21^{\prime}3}\equiv {\rm Tr}\left[{B}\left(-t_3\right){B}\left(-t_1\right){B}\left(-t_2\right)\rho\left(-\infty\right)\right]=
 {\rm Tr}\left[{B}\left(t_2\right){B}\left(t^{\prime}_1\right){B}\left(t_3\right)\rho\left(0\right)\right]$ (where $t^{\prime}_1\equiv t_3+t_2-t_1>t_3$ and time translational invariance of the system has been used) are identical due to T-symmetry of the system, which is consistent with Theorem~\ref{theorem_T}.

\subsection*{Generalized FDT for high-order correlations}

For time-translation symmetric systems in thermal equilibrium, the FDT has been established between the second-order correlations $C^{++}_{\sigma,2}$ and $C^{+-}_{\sigma,2}$. The FDT-like relations for higher order correlations have been studied only for classical systems~\cite{ PhysRevLett.87.040601, PhysRevA.86.043818} and for a special class of OTOCs in quantum systems~\cite{PhysRevE.97.012101}. Notwithstanding the importance of high-order correlations in quantum many-body physics, the formulation of FDT-like relations between them remains an open problem~\cite{RevModPhys.88.045008}.

For time-translation symmetric systems, we define the spectra of the Wightman correlations and the physical ones as
\begin{subequations}
\label{eq:spectra}
\begin{align}
  \tilde{W}_{\sigma,n}\left(\boldsymbol{\omega}\right) & \equiv \int^{+\infty}_{-\infty}  W_{\sigma,n} \prod^{n-1}_{j=1} e^{i \omega_j \tau_j} d\tau_j ,
  \label{eq:spectra_W}\\
     \tilde{C}^{{\boldsymbol \eta}}_{\sigma,n}\left(\boldsymbol{\omega}\right) & \equiv \int^{+\infty}_{-\infty} C^{{\boldsymbol \eta}}_{\sigma, n} \prod^{n-1}_{j=1} e^{i \omega_j \tau_j} d\tau_j, 
     \label{eq:spectra_C}
\end{align}
\end{subequations}
where $\boldsymbol{\omega}\equiv \left(\omega_{n-1},\ldots,\omega_2,\omega_1\right)$ and $\tau_j\equiv t_{j+1}-t_j$. In the following, we assume that the system is in a thermal equilibrium state  $\hat{\rho}=Z^{-1}e^{-\beta \hat{H}}$ with $Z\equiv {\rm Tr}e^{-\beta\hat{H}}$. The relations for systems in a general stationary state are established in Methods~\ref{App_sec:gFDTn}.

The generalized FDT for high-order correlation spectra is summarized in the following theorem. 
\begin{theorem}\label{theorem_TTI}
    When the system is time-translation invariant (i.e., its Hamiltonian is time-independent) and initially in a thermal equilibrium state, the correlation spectra of $n$-th order Wightman correlations satisfy the following cyclic relation (for $n>1$)
\begin{subequations}
\begin{equation}
\label{eq:theorem_TTI}
\tilde{W}_{\sigma\left(n\right),\ldots,\sigma\left(2\right),
\sigma\left(1\right)}\left(\boldsymbol{\omega}\right)=e^{\beta \omega_{\sigma\left(1\right)}-\beta \omega_{\sigma\left(1\right)-1}}
\tilde{W}_{\sigma\left(1\right),\sigma\left(n\right),\ldots,\sigma\left(2\right)}\left({\boldsymbol{\omega}}\right),
\end{equation}
with $\omega_0\equiv \omega_n\equiv 0$, and the spectra of the high-order physical correlations satisfy the following relation
    \begin{equation}\label{eq:FDTn1} 
        \tilde{C}^{+\eta_{n-1}\cdots\eta_{2}+}_{\sigma,n}\left({\boldsymbol{\omega}}\right)=\frac{i}{2} \coth \frac{\beta \omega_{\sigma(1)}-\beta \omega_{\sigma(1)-1}}{2} \tilde{C}^{+\eta_{n-1}\cdots\eta_{2}-}_{\sigma,n} \left({\boldsymbol{\omega}}\right).
    \end{equation}
  \end{subequations}
\end{theorem}
\begin{proof}
For a system in a thermal equilibrium state, we have
\begin{align}
\tilde{W}_{\sigma\left(n\right),\ldots,\sigma\left(2\right),\sigma\left(1\right)}\left({\boldsymbol{\omega}}\right)=&\int^{+\infty}_{-\infty}  \operatorname{Tr}\left[\hat{B}_{\sigma\left(n\right)}  \cdots \hat{B}_{\sigma\left(2\right)} e^{-\beta \hat{H}} e^{\beta \hat{H}}\hat{B}_{\sigma\left(1\right)}e^{-\beta \hat{H}}/Z\right] \prod^{n-1}_{j=1}e^{i \omega_j \tau_j} d\tau_j \notag \\
    &= \int^{+\infty}_{-\infty}  \operatorname{Tr}\left[\hat{B}_{\sigma\left(1\right)}\left(t_{\sigma\left(1\right)}-i\beta \right)\hat{B}_{\sigma\left(n\right)}  \cdots \hat{B}_{\sigma\left(2\right)} e^{-\beta \hat{H}}/Z \right] \prod^{n-1}_{j=1} e^{i \omega_j \tau_j} d\tau_j \notag \\
    &=e^{\beta \omega_{\sigma\left(1\right)}-\beta\omega_{{\sigma\left(1\right)}-1}}\tilde{W}_{\sigma\left(1\right),\sigma\left(n\right),\ldots,\sigma\left(2\right)}(\boldsymbol{\omega}),\notag
\end{align}
where $e^{\beta \hat{H}}\hat{B}_{{\sigma\left(1\right)}}\left(t_{\sigma\left(1\right)} \right)e^{-\beta \hat{H}}=\hat{B}_{{\sigma\left(1\right)}}\left(t_{\sigma\left(1\right)}-i\beta \right)$ has been used and to obtain the last equation the integration variables $\tau_{\sigma\left(1\right)}$ and $\tau_{{\sigma\left(1\right)}-1}$ have been changed to be $\tau_{\sigma\left(1\right)}'\equiv t_{{\sigma\left(1\right)}+1}-\left(t_{\sigma\left(1\right)}-i\beta\right)$ and $\tau_{{\sigma\left(1\right)}-1}'\equiv \left(t_{\sigma\left(1\right)}-i\beta\right)-t_{{\sigma\left(1\right)}-1}$, respectively, with $t_0\equiv 0$, which leads to the factor $e^{\beta \omega_{\sigma\left(1\right)}-\beta\omega_{{\sigma\left(1\right)}-1}}$ through $e^{i\omega_{\sigma\left(1\right)} \tau_{\sigma\left(1\right)}}=e^{i\omega_{\sigma\left(1\right)} \tau'_1}e^{\beta \omega_{\sigma\left(1\right)}}$ and $e^{i\omega_{{\sigma\left(1\right)}-1} \tau_{{\sigma\left(1\right)}-1}}=e^{i\omega_{{\sigma\left(1\right)}-1} \tau'_{{\sigma\left(1\right)}-1}}e^{-\beta \omega_{{\sigma\left(1\right)}-1}}$~\cite{kubo1966fluctuation}.
Using 
\begin{align}
2\tilde{C}^{+\eta_{n-1}\cdots\eta_{2}+}_{\sigma,n} & ={\rm Tr}\left[{\mathbb B}_{\sigma\left(n\right)}^+{\mathbb B}_{\sigma\left(n-1\right)}^{\eta_{n-1}}\cdots{\mathbb B}_{\sigma\left(2\right)}^{\eta_2}\hat{B}_{\sigma\left(1\right)}\hat{\rho}\right]+{\rm Tr}\left[\hat{B}_{\sigma\left(1\right)}{\mathbb B}_{\sigma\left(n\right)}^+{\mathbb B}_{\sigma\left(n-1\right)}^{\eta_{n-1}}\cdots{\mathbb B}_{\sigma\left(2\right)}^{\eta_2}\hat{\rho}\right],
\notag \\
i\tilde{C}^{+\eta_{n-1}\cdots\eta_{2}-}_{\sigma,n} & ={\rm Tr}\left[{\mathbb B}_{\sigma\left(n\right)}^+{\mathbb B}_{\sigma\left(n-1\right)}^{\eta_{n-1}}\cdots{\mathbb B}_{\sigma\left(2\right)}^{\eta_2}\hat{B}_{\sigma\left(1\right)}\hat{\rho}\right]- {\rm Tr}\left[\hat{B}_{\sigma\left(1\right)}{\mathbb B}_{\sigma\left(n\right)}^+{\mathbb B}_{\sigma\left(n-1\right)}^{\eta_{n-1}}\cdots{\mathbb B}_{\sigma\left(2\right)}^{\eta_2}\hat{\rho}\right],
\notag 
\end{align} 
and
Eq.~\eqref{eq:theorem_TTI}, we obtain Eq.~\eqref{eq:FDTn1}. 
\end{proof}

 The generalized FDT establishes connections between different types of correlation spectra, which can be used to simplify the measurement of correlation spectra. Repeatedly using Eq.~\eqref{eq:theorem_TTI}, we can show that among all Wightman correlations related by cyclic permutations only one is independent. So there are at most $(n-1)!$ independent $n$-th order Wightman correlations for time-translation symmetric systems. Furthermore, OTOCs are involved in Eq.~\eqref{eq:spectra} when $n>2$, so the generalized FDT enables us to obtain the spectra of OTOCs from the spectra of CTOCs (See Methods~\ref{OfC}). It can be shown that Eq.~(\ref{eq:FDTn1}) reduces to the high-order
fluctuation–dissipation relations derived in Ref.~\cite{PhysRevA.86.043818} for classical systems and the FDT for a special class of OTOCs established in~\cite{PhysRevE.97.012101} for quantum systems constitutes a special case of our results. For second-order correlations, the generalized FDT is reduced to the well-known Kubo formula~\cite{kubo1966fluctuation}. 

The generalized FDT beyond the second order cannot be directly observed in QNS due to the involvement of OTOCs in the definition of the correlation spectra. However, the relations between OTOCs and CTOCs induced by T- or S-symmetry guarantee that all third-order OTOCs in T- or S-symmetric systems  can be extracted from QNS. Therefore, the generalized FDT for third-order correlations can be established for CTOCs when the system has T- or S-symmetry (See Methods~\ref{T FDT} for details).

In the third order, we have the relations $\tilde{W}_{321}=e^{\beta \omega_1}\tilde{W}_{132}=e^{\beta \omega_2}\tilde{W}_{213}$ and $\tilde{W}_{123}=e^{-\beta \omega_2}\tilde{W}_{312}=e^{-\beta \omega_1}\tilde{W}_{231}$, and the generalized FDT for the physical third-order correlations as
\begin{align}
    \tilde{C} ^{+\eta+}_{\sigma,3}\left(\boldsymbol{\omega}\right)=&{i}{2} \coth \frac{\beta \omega_{\sigma(1)}-\beta \omega_{\sigma(1)-1}}{2}  \tilde{C} ^{+\eta -}_{\sigma,3}\left(\boldsymbol{\omega}\right), \label{eq:FDT31} 
\end{align}
which is the straightforward generalization of the second-order FDT. In Methods~\ref{spec of TFIM}, we calculate the correlation spectra for the example TFIM and show that they satify the generalized FDT.

\section*{Conclusion}

We have discovered structures in high-order quantum correlations arising from non-spatial symmetries including the generalized C-, T-, and S-symmetries and time-translational symmetry. The generalized C-symmetry induces selection rules for the QNS of the higher-order correlations, and the generalized T- and S-symmetries connect some types of rank-$k$ OTOCs to rank-$(k\pm 1)$ OTOCS, and in particular, rank-2 OTCOs to CTOCs.  The connections between different ranks of OTOCs make it possible to measure certain types (but not all) of rank-2 OTOCs using QNS without the requirement of (unphysical) time-arrow reversal. They also reduce some higher-rank OTOCs to lower ones. We have also established the generalized FDT between high-order quantum correlations for time translation invariant systems. These generalized relations reveal profound connections among different types of quantum correlations, and enable us to obtain the spectra of certain types and ranks of OTOCs from those of CTOCs. Since the exponential growth of OTOCs signifies information scrambling in quantum many-body systems~\cite{hosur2016chaos, PhysRevX.8.021013, PhysRevLett.126.030601}, the connection between CTOCs and OTOCs suggests that information scrambling in quantum many-body systems is largely constrained by T- or S-symmetry, or CTOCs can exhibit diverging behaviors in T- or S-symmetric systems, which is worth further exploration.

\newpage
\clearpage

\beginmethod
\section*{Methods}

\pagestyle{plain}

\subsection{Generalized relations for high-order quantum correlations due to time translation invariance.}\label{App_sec:gFDTn}
For a general stationary system, the state of the system satisfies $ i\partial_t \hat{\rho}(t) =\left[\hat{H},\hat{\rho}(t)\right]=0$, so the initial state must be a function of the Hamiltonian, $\rho\left(\hat{H}\right)$. With the Laplacian transformation, the state can be rewritten as $\rho\left(\hat{H}\right)=\int_{0}^{+ \infty}e^{-\beta \hat{H}}f\left(\beta\right) d\beta$. Then the spectra of the correlations satisfy
\begin{align}\label{eq:gTT1}
&\left.\tilde{W}_{\sigma\left(n\right),\ldots,\sigma\left(2\right),\sigma\left(1\right)}\left(\boldsymbol{\omega}\right)\right|_{{\rho}\left(\hat{H}\right)}  \equiv \int^{+\infty}_{-\infty}  \operatorname{Tr}\left[\hat{B}_{\sigma\left(n\right)}  \cdots \hat{B}_{\sigma\left(2\right)} \hat{B}_{\sigma \left(1\right)}{\rho}\left(\hat{H}\right)\right] \prod^{n-1}_{j=1}d\tau_j e^{i \omega_j \tau_j} \nonumber \\ 
&\phantom{\tilde{W}_{\sigma\left(n\right),\sigma\left(1\right)}}= \int_0^{+\infty}f\left(\beta\right)d\beta \int^{+\infty}_{-\infty}   \operatorname{Tr}\left[ \hat{B}_{\sigma\left(n\right)}  \cdots \hat{B}_{\sigma\left(2\right)} \hat{B}_{\sigma \left(1\right)} e^{-\beta\hat{H}}\right] \prod^{n-1}_{j=1}d\tau_j e^{i \omega_j \tau_j} \nonumber \\ 
&\phantom{\tilde{W}_{\sigma\left(n\right),\sigma\left(1\right)}}= \int_0^{+\infty}f\left(\beta\right)d\beta \int^{+\infty}_{-\infty}  e^{\beta\left(\omega_{\sigma\left(1\right)}-\omega_{{\sigma\left(1\right)}-1}\right)} \operatorname{Tr}\left[\hat{B}_{\sigma \left(1\right)} \hat{B}_{\sigma\left(n\right)}  \cdots \hat{B}_{\sigma\left(2\right)} e^{-\beta\hat{H}}\right] \prod^{n-1}_{j=1}d\tau_j e^{i \omega_j \tau_j} \nonumber \\ 
 & =  \left.\tilde{W}_{\sigma\left(1\right),\sigma\left(n\right),\ldots,\sigma\left(2\right)}\left(\boldsymbol{\omega}\right)\right|_{{\rho}\left(\hat{H}-\omega_{\sigma\left(1\right)}+\omega_{\sigma\left(1\right)-1}\right)}.
\end{align}
Note that the frequency-shifted density operator
$$   \rho\left(\hat{H}-\omega_{\sigma\left(1\right)}+\omega_{{\sigma\left(1\right)}-1}\right) \equiv  \int_{0}^{+ \infty}e^{-\beta \left(\hat{H}-\omega_{\sigma\left(1\right)}+\omega_{{\sigma\left(1\right)}-1}\right)}f\left(\beta\right) d \beta $$
is not necessarily normalized.

\subsection{Obtaining the spectra of OTOCs from those of CTOCs}\label{OfC}

As an example, we take the relation in Eq.~\eqref{eq:theorem_TTI} to obtain the spectra of third-order OTOCs from those of CTOCs.  
 The Wightman correlation $W_{abc}$ is an OTOC when $t_b<t_a$ and $t_c$.  We define the OTOC and CTOC spectra of $W_{abc}$ as
\begin{subequations}
\begin{align}
    \tilde{W}_{abc}^{\rm OTOC} \left(\boldsymbol{\omega}\right) & \equiv\iint_{t_b< t_a, t_c}  e^{i \omega_1 \tau_1} e^{i \omega_2 \tau_2} W_{abc} d\tau_1 d\tau_2, \label{eq:OTOC_Wabc}\\
    \tilde{W}_{abc}^{\rm CTOC} \left(\boldsymbol{\omega}\right) & \equiv \tilde{W}_{abc}\left(\boldsymbol{\omega}\right)-\tilde{W}_{abc}^{\rm OTOC}\left(\boldsymbol{\omega}\right).
\end{align}
\end{subequations}
Using $\tilde{W}_{abc}\left(\boldsymbol{\omega}\right)=e^{\beta \left(\omega_c-\omega_{c-1}\right)}\tilde{W}_{cab}\left(\boldsymbol{\omega}\right)$, we obtain
\begin{align}\label{T-C TOC}
    \tilde{W}_{abc}^{\rm OTOC}\left(\boldsymbol{\omega}\right)-e^{\beta \left(\omega_c-\omega_{c-1}\right)} \tilde{W}_{cab}^{\rm OTOC}\left(\boldsymbol{\omega}\right)=e^{\beta \left(\omega_c-\omega_{c-1}\right)}\tilde{W}_{cab}^{\rm CTOC}-\tilde{W}_{abc}^{\rm CTOC}.
\end{align}
The definition of OTOC and CTOC spectra can be generalized to higher orders, and some (but not all) of the OTOC spectra can be obtained using the CTOC ones.

\subsection{ FDT for third-order CTOCs in  T- or S-symmetric systems}\label{T FDT}
In T- or S-symmetric symmetric systems, all third-order OTOCs  can be obtained from CTOCs. Therefore, we can establish a FDT for third-order CTOCs in T- or S-symmetric systems.

 When the system is T-symmetric, we can obtain all OTOCs from CTOCs through
$    W_{abc}  =\beta_a\beta_b\beta_c  W_{cba}\left(\{t_j\rightarrow -t_j\}\right)$ as given by Theorem~\ref{theorem_T},
where $t_b<t_a,t_c$.
Then the OTOC spectra can be obtained from the CTOCs for T-symmetric systems
\begin{align}
    \tilde{W}_{abc}^{\rm OTOC}\left({\boldsymbol\omega}\right) & = \beta_a\beta_b\beta_c \iint_{t_b> t_a,t_c} d\tau_2  d\tau_1  e^{-i \omega_1 \tau_1} e^{-i \omega_2 \tau_2} W_{cba},\nonumber
\end{align}
where the r.h.s. of the equation involves only the CTOC $W_{123}$ with $t_2>t_1,t_3$. Similarly, for S-symmetric systems, using $    W_{abc}  =\beta_a\beta_b\beta_c  W_{abc}\left(\{t_j\rightarrow -t_j\}\right)$ as given by Theorem~\ref{theorem_S} and Eq.~\eqref{eq:OTOC_Wabc}, we get
\begin{align}
    \tilde{W}_{abc}^{\rm OTOC}\left({\boldsymbol\omega}\right) & = \gamma_a\gamma_b\gamma_c \iint_{t_b > t_a,t_c} d\tau_2  d\tau_1  e^{-i \omega_1 \tau_1} e^{-i \omega_2 \tau_2} W_{abc}.\nonumber
\end{align}
Therefore, the spectra on both sides of Eq. (\ref{T-C TOC}) can be obtained from CTOCs and the FDT for third-order CTOCs can be established.

\subsection{Third-order correlations of TFIM}\label{spec of TFIM}
Using the Jordan-Wigner transformation, we map the TFIM in Eq.~\eqref{eq:TFIM} to a free fermion model with the Hamiltonian  $\hat{H}=\sum_k \varepsilon_k\left(\hat{b}_k^{\dagger} \hat{b}_k-1 / 2\right)$, where $\varepsilon_k=2 \sqrt{1-2 \lambda \cos k+\lambda^2}$, and $\hat{b}^{\dagger}_k$ and $\hat{b}_k$ are annihilation and creation operators of fermions with wavevector $k$~\cite{mattis2006theory}.
The measurement operator can be rewritten as~\cite{chen2013quantum}
$$ \hat{B}=\frac{2}{\sqrt{N}}\sum_{k \geq 0} \left[\cos \left(2 \theta_k\right)\left(\hat{b}_k^{\dagger} \hat{b}_k+\hat{b}_{-k}^{\dagger} \hat{b}_{-k}-1\right)\notag - \mathrm{i} \sin \left(2 \theta_k\right)\left(\hat{b}_{-k} \hat{b}_k+\hat{b}_{-k}^{\dagger} \hat{b}_k^{\dagger}\right)\right],
$$ 
where $\theta_k$ is defined with $\tan \left(2 \theta_k\right)=\sin k/ \left(\cos k - \lambda\right)$. For a thermal equilibrium state, we obtain the third-order correlations as ${C}^{++-}_{321}=0$ and
\begin{subequations}
\begin{align}
{C}^{+--}_{321}=& -\frac{2^5}{{N}^{3/2}} \sum_{k \geq 0} \nu_k \lambda_k^- 
\left\{\cos\left(2\varepsilon_k \tau_1\right)-\cos\left[2\varepsilon_k\left(\tau_1+\tau_2\right)\right]\right\}, \\
{C}^{+-+}_{321}=& \frac{2^4} {{N}^{3/2}}\sum_{k>0}\nu_k \lambda_k^+\left\{\sin\left(2\varepsilon_k\tau_1\right)+\sin\left(2\varepsilon_k\tau_2\right)-\sin\left[2\varepsilon_k\left(\tau_1+\tau_2 \right)\right]\right\}, \\
{C}^{+++}_{321}=& \frac{2^3}{{N}^{3/2}}\sum_{k>0}\lambda_k^-\left[-\cos^3(2\theta_k)-\nu_k\cos\left(2\varepsilon_k\tau_2\right)\right],
\end{align}
\end{subequations}
where $\nu_k= \cos\left(2\theta_k\right)\sin^2\left(2 \theta_k\right)$, and $\lambda^{\pm}_k=\frac{1\pm e^{-2\beta \varepsilon_k}}{\left(1+e^{-\beta \varepsilon_k}\right)^2}$.
Their spectra are $\tilde{C}^{+--}_{321}\left(\boldsymbol{\omega}\right)=0$ and 
\begin{subequations}
\begin{align}
    \tilde{C}^{+--}_{321}\left(\boldsymbol{\omega}\right)=&-\frac{2^6  \pi ^2}{{N}^{3/2}} \sum_{k \geq 0} \nu_k \lambda_k^- \left\{ \left[\delta\left(
    \omega_1-2\epsilon_k\right)+\delta\left(
    \omega_1+2\epsilon_k\right)\right]\delta\left(
    \omega_2\right) \right.\notag \\
    &\left.+\delta\left(
    \omega_2+2\epsilon_k\right)\delta\left(
    \omega_1+2\epsilon_k\right)+\delta\left(
    \omega_2-2\epsilon_k\right)\delta\left(
    \omega_1-2\epsilon_k\right) \right\} ,        \\
    \tilde{C}^{+-+}_{321}\left(\boldsymbol{\omega}\right)=&\frac{i 2^5 \pi^2}{{N}^{3/2}}\sum_{k>0}\nu_k \lambda_k^+  \left\{ \left[\delta\left(
    \omega_1-2\epsilon_k\right)-\delta\left(
    \omega_1+2\epsilon_k\right)\right]\delta\left(
    \omega_2\right)\right.\notag \\
    &+ \left[\delta\left(
    \omega_2-2\epsilon_k\right)-\delta\left(
    \omega_2+2\epsilon_k\right)\right]\delta\left(
    \omega_1\right)\notag  \\
    &\left.+\delta\left(
    \omega_2+2\epsilon_k\right)\delta\left(
    \omega_1+2\epsilon_k\right)-\delta\left(
    \omega_2-2\epsilon_k\right)\delta\left(
    \omega_1-2\epsilon_k\right) \right\}, \\
    \tilde{C}^{+++}_{321}\left(\boldsymbol{\omega}\right)=& \frac{2^4 \pi ^2}{{N}^{3/2}}\sum_{k>0} \lambda_k^-\left\{-\cos^3\left(2\theta_k\right) \delta\left(\omega_1\right)\delta\left(\omega_2\right)\right.\notag \\
    & \left.-\nu_k  \delta\left(\omega_1\right)\left[\delta\left(\omega_2-2 \epsilon_k\right)+\delta\left(\omega_2+2 \epsilon_k\right)\right]\right\}. \end{align}
\end{subequations}
 We find that these correlation spectra satisfy
\begin{align}
    \tilde{C} ^{+- +}_{321}\left(\boldsymbol{\omega}\right)=&\frac{i}{2}\coth\frac{\beta \omega_1}{2} \tilde{C} ^{+-  -}_{321}\left(\boldsymbol{\omega}\right),
\end{align}
and $\tilde{C} ^{++  -}_{321}\left(\boldsymbol{\omega}\right)=-2i\tanh\frac{\beta \omega_1}{2}  \tilde{C} ^{+++}_{321}\left(\boldsymbol{\omega}\right)=0$, as determined by the generalized FDT.

\bibliography{refs}

\section*{Acknowledgments}
\paragraph*{Funding:}
This work was supported by the Quantum Science and Technology - National Science and Technology Major Project No. 2023ZD0300600, the National Natural Science Foundation of China/Hong Kong Research Council Collaborative Research Scheme Project CRS-CUHK401/22, the Hong Kong Research Grants Council Senior Research Fellow Scheme Project SRFS2223-4S01, and the New Cornerstone Science Foundation.
\paragraph*{Competing interests:}
There are no competing interests to declare.

\end{document}